\newif\iffull

\fulltrue

\documentclass[sigconf, nonacm=true,10pt]{acmart}

\usepackage{amssymb,amsfonts,amsmath,amsthm, bbm}
\usepackage{fullpage}
\usepackage{tikz}
\usepackage{graphicx}
\usepackage{verbatim}
\usepackage{array}
\usepackage{multirow}
\usepackage{latexsym}
\usepackage{paralist}
\usepackage{enumitem}
  \setlist[itemize]{leftmargin=*}
  \setlist[enumerate]{leftmargin=*}
\usepackage[capitalise]{cleveref}
\usepackage{dsfont}
\usepackage{url}
\usepackage{braket}
\usepackage{mathrsfs}
\usepackage{color}
\usepackage{mdframed}
\usepackage{mathtools}

\newtheorem{theorem}{Theorem}
\newtheorem{corollary}{Corollary}[section]

\newtheorem{definition}[corollary]{Definition}
\newtheorem{lemma}[corollary]{Lemma}
\newtheorem{claim}[corollary]{Claim}

\newtheorem{fact}[corollary]{Fact}

\theoremstyle{definition}

\newcommand{\R}{{\mathbb R}}

 \newcommand{\abs}[1]{\lvert #1 \rvert}

\newcommand{\seq}{\subseteq}

\newcommand{\Bits}{\{0,1\}}
\newcommand{\A}{\mathcal A}
\newcommand{\E}{\mathbb E}
\newcommand{\parhead}[1]{\noindent {\textbf{#1}} \hskip 0.6em}
\newcommand{\B}{\mathcal B}
\newcommand{\reqseq}{(i_t)_{t=1}^{T}}
\newcommand{\opt}{\mathsf{OPT}}

\begin{document}

\title{Lower Bounds for Caching with Delayed Hits}

\author{Peter Manohar}
\affiliation{%
  \institution{Carnegie Mellon University}
 }
\email{pmanohar@cs.cmu.edu}

\author{Jalani Williams}
\affiliation{%
  \institution{Carnegie Mellon University}
}
\email{jalaniw@cs.cmu.edu}


\begin{abstract}
Caches are a fundamental component of latency-sensitive computer systems. Recent work of \cite{AtreSWB20} has initiated the study of delayed hits: a phenomenon in caches that occurs when the latency between the cache and backing store is much larger than the time between new requests. We present two results for the delayed hits caching model.

\emph{(1) Competitive ratio lower bound.} We prove that the competitive ratio of the algorithm in \cite{AtreSWB20}, and more generally of any deterministic online algorithm for delayed hits, is at least $\Omega(k Z)$, where $k$ is the cache size and $Z$ is the delay parameter.

\emph{(2) Antimonotonicity of the delayed hits latency.} Antimonotonicity is a naturally desirable property of cache latency: having a cache hit instead of a cache miss should result in lower overall latency. We prove that the latency of the delayed hits model is not antimonotone by exhibiting a scenario where having a cache hit instead of a miss results in an \emph{increase} in overall latency. We additionally present a modification of the delayed hits model that makes the latency antimonotone.
\end{abstract}

\maketitle



\section{Introduction}
Caches are a key component of real-world computer systems, improving throughput for applications that access the same data frequently. Caches serve as an intermediary between a client requesting items and a backing store containing the items, masking the long delay to fetch an item from the backing store by storing a small number of items locally. 

In the classical caching problem, we are given $n$ items, and a cache containing a subset $S \seq [n]$ of $k$ of these items. At every timestep, the cache gets a request for an item $i \in [n]$. If $i \in S$, we say that the cache has a ``hit''; otherwise, the cache has a ``miss'' and the item is retrieved from the backing store. The caching algorithm then decides whether or not to cache the newly retrieved item, and if so what item to evict from the cache, with the ultimate goal being to minimize the total number of misses. There are many different classical caching algorithms, for example the Least Recently Used (LRU) policy, which discards the item that was requested the furthest in the past.

The classical caching problem is a theoretical model of real-world caches that assumes an item appears immediately in the cache once it is requested. This assumption is reasonable when the time between requests is much slower than the time it takes to fetch the requested item and load it into the cache. This is because when the next request is received, the item that was previously being fetched is already loaded into the cache. However, when this is not the case several requests can arrive while an item is being fetched, producing a phenomenon known as \emph{delayed hits}.  A delayed hit occurs when multiple requests for the same item occur while the item is already being fetched.

To understand what a delayed hit is, consider the following example. Suppose that the cache $S$ initially consists of the items $\{1,2\}$, and that it takes $Z = 100$ ms for an item to be retrieved from the backing store. At time $t = 0$ ms, a request for item $3$ arrives, and at time $t = 25$ ms and $t = 50$ ms two more requests for item $3$ arrive. Since $3$ is not in the cache, the first request misses, and experiences a delay of $100$ ms, the time it takes for the item to be retrieved. However, the second and third requests only experience delays of $75$ ms and $50$ ms respectively, as a request for item $3$ was already ``in flight'' when the other requests arrived.

Delayed hits are not a minor technical issue with the classical caching model: they contribute substantially to actual latencies in practice. \cite{AtreSWB20} devised an algorithm for caching with delayed hits that had between $0.1\%$ and $38\%$ better latency compared to the best classical caching algorithm. This is especially significant because the algorithm in \cite{AtreSWB20} is \emph{online}: the algorithm's decisions are only based on the past requests, whereas the optimal classical caching algorithm is \emph{offline}: it is given the full sequence of requests in advance, and so its decisions can not only depend on past requests, but also future ones.

Understanding delayed hits is thus essential to minimizing cache latency in practice. However, so far there has been little work on delayed hits in practice, and even less work on achieving a theoretical understanding of delayed hits. In this paper, we prove two new results about delayed hits, making progress towards a better theoretical understanding of delayed hits.

\subsection{Our results}

We now describe our two main results for the delayed hits caching problem.

\parhead{Competitive ratio lower bound.} In our first result, we show that the theoretical guarantees of the algorithm in \cite{AtreSWB20} are actually quite poor. Specifically, we lower bound the \emph{competitive ratio} of the algorithm: the smallest $\alpha$ such that the latency of the online algorithm on any sequence of requests is at most $\alpha$ times the latency of the best offline algorithm on that sequence. The competitive ratio is the value used to judge the quality of online algorithms, and captures a notion of minimal regret. An $\alpha = O(1)$ indicates that the online algorithm is always a constant-factor approximation of the optimal offline algorithm, which is typically quite good, whereas an $\alpha$ that grows asymptotically usually indicates poor performance. We prove that the competitive ratio of \cite{AtreSWB20} is at least $\Omega(kZ)$, where $k$ is the size of the cache and $Z$ is the time it takes to load an item into the cache. More generally, we prove the following theorem.
\begin{theorem}
\label{thm:1}
Any deterministic algorithm $\A$ for the delayed hits problem has a competitive ratio of \mbox{$\alpha_{\A} \geq \Omega(k Z)$}. 
\end{theorem}
The key idea is to use the fact that $\A$ is deterministic to construct a fixed sequence of cache requests where $\A$ has a cache miss on every request, whereas the offline optimal algorithm only has one miss. The proof of \cref{thm:1} can be found in \cref{sec:ratio}.

\parhead{Non-antimonotonicity of latency.} For a sequence of $T$ cache requests, the performance of any algorithm can be encoded as a sequence of bits $b_1, \dots, b_T$, denoting whether or not the $i$-th request was a full cache hit.\footnote{Delayed hits are viewed as ``partial misses'' and do not count as hits.} The total latency of the algorithm can be computed from these bits, so the latency is $\ell(b_1, \dots, b_T)$ for some latency function $\ell \colon \Bits^T \to \R$. Note that not all settings of the $b_i$'s correspond to valid caching algorithms, e.g.\ setting $b_i = 1$ for every $i$ is typically not valid, as this would imply that every request was a cache hit. Intuitively, $\ell$ should be antimonotone\footnote{A boolean function $f \colon \Bits^n \to \R$ is antimonotone if for every $b,b' \in \Bits^n$ where $b'_i \geq b_i$ for all $i$ it holds that $f(b') \leq f(b)$.} as having \emph{more} cache hits should only be able to \emph{decrease} the total latency. We show that, surprisingly, this is not the case by proving the following theorem.
\begin{theorem}\label{thm:2}
There exists a sequence of cache requests such that the delayed hits latency function $\ell$ is not antimonotone.
\end{theorem}
In particular, we exhibit a scenario where an algorithm can choose between having a request hit or having it miss, without changing whether or not the other requests hit or miss. In this scenario, we show that not only is it better for the algorithm to have the request \emph{miss}, but moreover this choice is \emph{optimal}, and is the unique way to minimize latency. Our key idea here is to exploit the fact that having a cache miss can decrease the latency of later requests to design a request sequence gadget where having a cache miss results in an overall decrease in latency. The proof of \cref{thm:2} can be found in \cref{sec:antimonotone}.

We then exhibit a model of delayed hits different from the one in \cite{AtreSWB20} that we call ``antimonotone delayed hits'', and show that the latency function for this model is always antimonotone. We then give the following reduction from antimonotone delayed hits to delayed hits.
\begin{theorem}\label{thm:3}
Any algorithm $\A$ with cache size $k$ for antimonotone delayed hits can be transformed to an algorithm $\B$ with cache size $k + Z$ for delayed hits such that $\text{latency}(\B) \leq \text{latency}(\A)$ for every sequence of cache requests.
\end{theorem}
The key idea is to modify the delayed hits model so that the strange scenario in \cref{thm:2} does not occur, and then show that \cref{thm:2} is essentially the \emph{only} way in which the delayed hits model can be non-antimonotone. The formulation of the antimonotone delayed hits model and the proof of \cref{thm:3} can be found in \cref{sec:model}.

We note that the non-antimonotonicity of latency is very different from Belady's anomaly \cite{BeladyNS69}. Belady's anomaly is the fact that for certain classical caching algorithms and request sequences, increasing the cache size can sometimes result in worse overall latency. This is qualitatively different from antimonotonicity because Belady's anomaly is a property that depends on both the cache size $k$ and the caching algorithm, whereas antimonotonicity is a property of the latency function $\ell(\cdot)$, and the function $\ell(\cdot)$ is the same for all algorithms and all cache sizes. For example, in the classical caching problem the latency function is simply $\ell(b_1, \dots, b_T) :=  \sum_{i = 1}^T(1 - b_i)$, which is clearly antimonotone and also algorithm/cache~size independent. The cache size $k$ will determine what $(b_1, \dots, b_T)$ can be realized by a caching algorithm, but it does not affect the latency function $\ell(\cdot)$.

\subsection{Prior work}
The classical caching problem has been studied extensively, in both the offline and the online setting.  In the offline setting, \cite{Belady66} showed that the optimal algorithm is very simple: evict the item $j \in S$ that is requested again the latest in the future. In the online setting, \cite{SleatorT84} showed that every deterministic algorithm has a competitive ratio of at least $k$, and that the Least Recently Used (LRU) algorithm has a competitive ratio of exactly $k$. \cite{FiatKLMSY91} gave a randomized online algorithm with a competitive ratio of $2 H_k$, and showed that no randomized online algorithm can have competitive ratio better than $H_k$, where $H_k := 1 + \frac{1}{2} + \dots + \frac{1}{k}$ is the $k$-th harmonic number. Shortly after, \cite{McGeochS91} gave a randomized algorithm with a competitive ratio of $H_k$, which matches the lower bound.

On the other hand, there is little prior work on delayed hits. The formal model for delayed hits caching was only recently introduced in \cite{AtreSWB20}. This paper formulated an (inefficient) algorithm for the offline delayed hits problem, and gave an online algorithm based on rounding an efficient relaxation of the offline algorithm. To demonstrate the effectiveness of their online algorithm, the authors implemented the algorithm and showed significant improvements in latency in practice compared to the classical caching offline optimal algorithm of \cite{Belady66}. However, the paper did not prove any theoretical guarantees about the algorithm.

\section{Preliminaries}
\label{sec:prelims}

\subsection{The delayed hits model}
The delayed hits model is very similar to the classical caching model. We first recall the classical caching model, and then explain the changes in the delayed hits model. 

In the classical caching model, there are $n$ items, and a cache $S$ of size $k$ containing a subset of the $n$ items. The cache is initialized to $S_0 := \{1, \dots, k\}$, and the model proceeds in discrete timesteps. At the $t$-th timestep, the cache is currently $S_{t-1}$ and an item $i_t \in [n]$ arrives. If $i_t \in S_{t-1}$ then we have a cache hit, and we set $S_{t} \gets S_{t-1}$. If $i_t \notin S_{t-1}$ then we have a cache miss, and the caching algorithm can either evict some $j_t$ from the cache and replace it with $i_t$, thus setting \mbox{$S_t \gets S_{t-1} \cup \{i_t\} \setminus \{j_t\}$}, or leave the cache unchanged and set $S_{t} \gets S_{t-1}$. In either case, the algorithm incurs a cost of $1$ for the miss. We can view each timestep here as having two phases: the request phase, where the item $i_t$ is requested, and the retrieval phase, where the item $i_t$ is returned from the backing store and the cache is updated.

In the delayed hits model, there are a few significant changes. First, we have the delay parameter $Z$, which is a positive integer representing the number of timesteps it takes to fetch an item from the backing store. Now, when we have a cache miss for $i_t$, the fetch for item $i_t$ terminates $Z - 1$ timesteps in the future. The quantity $Z$ in the model corresponds to the maximum number of requests that can arrive during one fetch operation in the physical system. Second, we have a set $Q$ containing all the requested items that have not yet been served. As before, each timestep has two phases. In the request phase, a request for item $i_t$ arrives. If $i_t \in S_{t-1}$ then we have a cache hit. Otherwise, the request is sent to the backing store and we update $Q \gets  Q + (i_t, t)$, where the $+$ denotes the append operation. In the retrieval phase, item $i_{t - Z + 1}$ is returned, but \emph{only if it was a cache miss during its request phase earlier}. All requests for $i_{t - Z + 1}$ in $Q$ are then served, so we remove all tuples $(i_{t - Z + 1}, t')$ from $Q$. For each tuple we remove, the total latency increases by $t - t' + 1$, as the item $i_{t - Z + 1}$ arrived during the request phase at time $t'$ and the request has finished being served during the retrieval phase at time $t$. This is the \emph{latency incurred} by the $t'$-th item. As in the classical caching model, the algorithm then can decide to either cache $i_{t - Z + 1}$ or not. Finally, we also allow the requested item $i_t$ to be $0 \notin [n]$, denoting that no item was requested during the $t$-th timestep. In this case, at time $t$ the request phase is skipped and we proceed directly to the retrieval phase, and at time $t + Z - 1$ the retrieval phase is skipped.

We demonstrate how the model works with a simple example. First, consider the case where $i_t = k+1$ for $t = 1, \dots, Z$. In this case, we will show that the total latency for \emph{any} caching algorithm is $Z (Z+1)/2$. Since the cache $S_0$ is always $\{1, \dots, k\}$, the first request $i_1$ misses. Moreover, since $i_1$ can only be added to the cache at the end of the retrieval phase at time $t = Z$, all requests $i_1, \dots, i_Z$ are cache misses. So, for $t = 1, \dots, Z$, the request phase simply adds the tuple $(k+1, t)$ to $Q$, resulting in $Q = ((k+1, 1), \dots, (k+1, Z))$ before the retrieval phase at time $t = Z$. At the retrieval phase for $t = Z$ the request for $i_1$ returns. At this point, all requests in $Q$ can be served, and we incur a cost of $(Z - 1 + 1) + (Z - 2 + 1) + \dots (Z - Z + 1) = \sum_{r = 1}^Z r = Z(Z + 1)/2$. Note that this explains the choice of having a request arrive $Z - 1$ timesteps in the future (as opposed to $Z$), as it means that when the request returns it ``covers'' exactly $Z$ requests.

We note that the delayed hits model for $Z = 1$ is identical to the classical caching model, and in general the smaller $Z$ is the closer delayed hits is to classical caching. In practice, the value of $Z$ can vary substantially, from $Z = 1$ all the way to $Z \approx 2 \times 10^{5}$ \cite{AtreSWB20}.

\parhead{Hits, delayed hits, and misses.} Each request $i_t$ will always incur a latency in $\{0, \dots, Z\}$. If the latency is $0$ then we say that $i_t$ is a \emph{hit}. If the latency is $Z$ then we say that $i_t$ is a \emph{miss}, and if the latency is in $\{1, \dots, Z-1\}$ then we say that $i_t$ is a \emph{delayed hit}. In the case where we only are differentiating between hits and misses we will treat delayed hits as misses. This is because a delayed hit $i_t$ was a cache miss in the request phase at time $t$, and so from the perspective of the cache in the model it is a miss. 

\subsection{Offline and online algorithms}
In this section, we formally define competitive analysis, as well as online and offline algorithms. We begin with a definition.
\begin{definition}
A request sequence $\sigma$ is a finite sequence $\sigma:= (i_t)_{t = 1}^T$ where $i_t \in [n] \cup \{0\}$. If $i_t \in [n]$ then we say that $i_t$ is requested at time $t$; if $i_t = 0$ then no item is requested at time~$t$.
\end{definition}

\parhead{Algorithmic setting.} We now describe the nature of an algorithm in the delayed hits setting. Suppose an algorithm is given a request sequence $(i_t)_{t=1}^T$. At each request $i_t$, the algorithm serves $i_t$ immediately if it is in the cache, and otherwise places it in the request queue $Q$. Once a request returns from the backing store at time $t$, the algorithm must decide which object $j_t$ (if any) in the cache $S_{t}$ will be evicted to make room for the retrieved request, $i_{t - Z + 1}$. Given, a request sequence $\sigma$, the output of a delayed hits algorithm is then precisely the sequence of chosen evictions $(j_t)_{t=1}^T$. We use the convention that $j_t = 0$ if the algorithm chose not to cache $i_{t - Z + 1}$, or if $i_{t - Z + 1} = 0$ (so that no request was returned from the backing store at time $t$).

\parhead{Feasibility.} A sequence of evictions $(j_t)_{t = 1}^T$ is \textit{feasible} for a request sequence $\sigma$ if every attempted eviction $j_t$ is feasible, i.e.\ if $j_t \in S_t$ for all $t$. Likewise, given a request sequence $(i_t)_{t = 1}^T$ and eviction sequence $(j_t)_{t = 1}^T$, one can reconstruct the cache state $S_t$ at every time step. We define a sequence of cache states $(S_0, S_1, \dots, S_T) = (S_t)_{t=0}^T$ to be feasible in a similar way.

\begin{definition}
 A delayed hits algorithm is an algorithm $\A$ that takes as input a  request sequence $\sigma = \reqseq$ and outputs a feasible sequence of cache states $(S_t)_{t = 0}^T$.
\end{definition}

\parhead{Online and offline algorithms.} With these basic concepts defined, we now distinguish between offline and online delayed hits algorithms.

\begin{definition}
An offline delayed hits algorithm is a delayed hits algorithm wherein the chosen eviction at time $t$, $j_t$, can depend upon the full request sequence $\sigma$; i.e., $j_t = f_t( \sigma )$ for some function $f_t$.
\end{definition}

By contrast, an online algorithm evicting at time $t$ may only make use of information available at time $t$. Formally, let $\sigma_{t'}$ be the truncation of the request sequence at time $t'$, i.e. $\sigma_{t'} = (i_t)_{t=1}^{t'}$. We define an online delayed hits algorithm as follows:

\begin{definition}
 An online delayed hits algorithm is a delayed hits algorithm wherein the chosen eviction at time $t$, $j_t$ depends only upon the history of requests and evictions. In other words, $j_{t'} = f_\ell( (j_t)_{t = 1}^{t'-1}, \sigma_t)$ for some function $f_{t'}$.
\end{definition}

We make no assumptions about the function $f_t$ here, besides that it computes $j_t$ in finite time. In fact, we allow for the possibility that the computation of $f_t$ might use random bits; if an algorithm $\A$ uses an outside source of randomness during its computation, we call $\A$ a \textit{randomized} algorithm.

\parhead{The competitive ratio.} The competitive ratio is a measure of how far the performance of the an online algorithm can deviate from that of the optimal offline algorithm. For a request sequence $\sigma$, we let $\opt(\sigma)$ denote the latency incurred by the optimal offline delayed hits algorithm, and let $\E[\A(\sigma)]$ be the expected latency incurred by the algorithm $\A$, where the expectation is over potential randomness used in the computation of the $f_t$'s. We define the competitive ratio as an asymptotic bound between these two quantities.
\begin{definition}
Let $\A$ be an online delayed hits algorithm. The competitive ratio of $\A$, denoted by $\alpha_{\A}$, is the smallest value of $\alpha_{\A}$ such that
\begin{equation*}
\E[\A(\sigma)] \leq \alpha_{\A} \opt(\sigma)
\end{equation*}
holds for every request sequence $\sigma$.
\end{definition}
We note that $\alpha_{\A} \geq 1$ always holds, since \mbox{$\E[\A(\sigma)] \geq \opt(\sigma)$} because $\opt$ is optimal.

\subsection{Latency functions for delayed hits}
In this section, we define the delayed hits latency function for a given request sequence $\sigma$ and show that the delayed hits latency function gives exactly the latency incurred by a delayed hits algorithm.

\begin{definition}
For a given request sequence $\sigma = (i_t)_{t = 1}^T$ and an algorithm $\A$, we define the hit sequence of the execution of $\A$ on $\sigma$ to be the vector $b \in \Bits^T$ where $b_t = 1$ if the request for $i_t$ was a hit in the execution, and $0$ otherwise. Delayed hits count as misses.
\end{definition}
We note that we can only define hit sequences with respect to the execution of $\A$ as $\A$ may not be deterministic.

Some hit sequences cannot be produced by any algorithm $\A$. For instance, if $k < n$ and we consider the request sequence $\sigma = (1, 2, \dots, n)$, then the hit sequence $(1, 1, 1, \dots, 1)$ cannot be produced by any $\A$, as any algorithm $\A$ must have at least $1$ cache miss.

\begin{definition}
A hit sequence $b$ is \emph{feasible} for a request sequence $\sigma$ if there is an algorithm $\A$ such that the execution of $\A$ on $\sigma$ produces the hit sequence $b$ with nonzero probability.
\end{definition}

We now define the delayed hits latency function and prove that it is well-defined.
\begin{lemma}[Latency function]
\label{lem:latencyfn}
For every request sequence $\sigma = (i_t)_{t = 1}^T$, there is a computable function $\ell_{\sigma} \colon \Bits^T \to \R$ such that for every algorithm $\A$, the following holds. For any execution of the algorithm $\A$ on $\sigma$, letting $b$ be the corresponding hit vector and $L$ be the total latency incurred during this execution, we have that $\ell_{\sigma}(b) = L$.
\end{lemma}
\begin{proof}
For each $i_t$, let $p_t := \min_{t' \in [t - Z + 1, t] : b_{t'} = 0 \wedge i_{t'} = i_t} t'$, and let $l_t := (1 - b_t)(Z - (t - p_t))$. The claim is that $L = \sum_{t = 1}^T l_t$. Since the $l_t$'s are clearly computable from $\sigma$ and $b$, the lemma follows.

It suffices to show that $l_t$ is precisely the latency incurred by $i_t$ in the execution of the algorithm. We observe that if $b_t = 1$ then $l_t = 0$, so it remains to argue that the latency incurred is $Z - (t - p_t)$ when $i_t$ is a miss. When $i_t$ is a miss, the only way the latency for $i_t$ to be less than $Z$ is if there is already a request for $i_t$ ``in flight'' at time $t$. Suppose that the request at time $t'$ is the request that, in the retrieval phase, is used to serve $i_t$. Since the request for $i_{t'}$ returns in the retrieval phase at time $t'' = t' + Z - 1$, the total latency incurred for $i_t$ is $t'' - t + 1 = t' + Z - 1 - t + 1 = Z - (t - t')$. Hence, finding the smallest such $t'$ results in the earliest time $t'$ that is used to serve the request $t$, which is the latency incurred in the algorithm. The value $p_t$ is the earliest such $t'$, as it is the earliest time $t'$ that $i_t$ was requested where the request was (1) a miss, so that the request will be sent to the backing store, and (2) $t' \leq t$ and $t - t' \leq Z - 1$, so that the request for $i_{t'}$ will come back in a retrieval phase $t''$ with $t'' \geq t$. This completes the proof.
\end{proof}

\subsection{Antimonotone boolean functions}
In this section, we define antimonotone boolean functions. We begin by defining a partial ordering on $\Bits^n$.

\begin{definition}
Let $b, b' \in \Bits^n$. We say that $b \leq b'$ if for every $i \in [n]$ it holds that $b_i \leq b'_i$.
\end{definition}
Note that this is only a partial ordering since, e.g., \mbox{$(1,0), (0,1) \in \Bits^2$} are incomparable with this relation. We now define monotone and antimonotone boolean functions.

\begin{definition}
A boolean function $f \colon \Bits^n \to \R$ is \emph{monotone} if for every $b,b' \in \Bits^n$ with $b \leq b'$ it holds that $f(b) \leq f(b')$. A boolean function $f \colon \Bits^n \to \R$ is \emph{antimonotone} if $-f$ is monotone.
\end{definition}
We note that both AND and OR are monotone, NAND and NOR are antimonotone, and XOR is neither monotone nor antimonotone.

\section{Lower bound on the competitive ratio}
\label{sec:ratio}
In this section, we prove \Cref{thm:1}. Given a deterministic caching algorithm $\A$, we show how to construct a request sequence $\sigma_\A$ where $\A$'s latency is a factor of $\Omega(k Z)$ larger than the latency of the optimal offline algorithm. We begin by defining our two building blocks: \emph{pure} and \emph{bursty} requests.

\parhead{Pure and bursty requests.}  A pure request is a sequence of requests of the form $(0^Z, i, 0^Z)$, where the notation $j^Z$ means that $j$ is requested $Z$ timesteps in a row. A bursty request is a sequence of the form $(0^Z, i^Z, 0^Z)$. We observe that pure and bursty requests are \emph{isolated}: there are no ``in flight'' requests when the item $i$ is first requested, and there are no ``in flight'' requests at the end of the sequence. Because of this, if the first item requested in either of these sequences is a cache hit then every item requested in the sequence is a cache hit, and likewise if the first item requested is a miss then every item requested is a miss. Because of this, we say that a pure/bursty request is a hit if the first item requested is a hit; else it is a miss. Clearly, if a pure/bursty request is a hit then the latency accrued is $0$. If a pure request is a miss then the latency accrued is $Z$, and if a bursty request is a miss then the latency accrued is $Z + (Z-1) + \dots + 1 = \frac{Z(Z+1)}{2}$.

\parhead{Constructing the request sequence $\sigma_{\A}$.} We now construct a request sequence $\sigma_{\A}$ iteratively from pure and bursty requests, using a marking procedure defined as follows. The marking procedure maintains a set of marked items $M$, initialized as $M = \emptyset$. Whenever item $i$ is requested as part of a bursty request, we add it to $M$, marking it. This mark persists outside of the cache, and can not be removed, e.g.\ if an object $i$ is marked and is requested again before the request sequence ends, then the mark persists unaltered.

Assume that $n > k$, and recall that the cache is initialized to $S_0 := \{1, \dots, k\}$. We construct $\sigma_{\A}$ using only pure/bursty requests. We define the request sequence $\sigma_{\A}$ iteratively as follows. We initialize $\sigma_{\A}$ to be a pure request for $k+1$. Then, we append a bursty request for the item in $\{1, \dots,k+1\}$ currently not in $\A$'s cache, and then afterwards we append a bursty request for the next item in $\{1, \dots, k+1\}$ currently not in $\A$'s cache, and so on. We terminate this process when $k$ items have been marked, that is when $\abs{M} = k$, and this results in the final request sequence $\sigma_{\A}$.

\parhead{Computing the latency of $\opt$.} We show that the optimal algorithm achieves a latency of $Z$. Every algorithm begins with $1, \dots, k$ in their cache, and the first request in $\sigma_\A$ is a pure request for object $k+1$. Thus, every algorithm must miss on the first request, so any algorithm must have a latency of at least $Z$ on $\sigma_\A$. We now show that there is a way to achieve a latency of $Z$, which makes $\opt(\sigma_{\A}) = Z$. Let $M$ be the set of all marked items at the end of $\sigma_{\A}$. Because we end when $k$ items have been marked, there is some item $j \in \{1, \dots, k+1\} \setminus M$, and by definition this item was never requested in any of the bursty requests in $\sigma_{\A}$. Hence, if $\opt$ evicts $j$ during the pure request for $k+1$\footnote{Note that we could have $j = k+1$, in which case $\opt$ simply doesn't cache $k+1$.} and never changes the cache afterwards, then $\opt$'s cache $S$ will contain every object in $\{1, \dots, k+1\}$ except for $j$. Since these are the only objects requested in the bursty requests of $\sigma_{\A}$, it follows that $\opt$ will never have another miss. Thus, $\opt(\sigma_\A) = Z$.

\parhead{Lower bounding the latency of $\A$.} We now show that $\A$'s latency at least $Z + k \frac{Z(Z+1)}{2}$. We observe that, by construction of $\sigma_{\A}$, the algorithm $\A$ misses on \emph{every} request in $\sigma_{\A}$. This is because the each item in $\sigma_{\A}$ is chosen to be precisely the item that is not in $\sigma_{\A}$'s cache at that time. Since we terminate the request sequence once $k$ items have been marked, it follows that $\sigma_{\A}$ contains at least $k$ bursty requests. Hence, the latency of $\A$ is at least
\begin{equation*}
\A(\sigma_{\A}) \geq Z + k \frac{Z(Z+1)}{2} \enspace.
\end{equation*}

\parhead{Putting it together.} Since $\A(\sigma_{\A}) \geq Z + k \frac{Z(Z+1)}{2} $ and $\opt(\A) = Z$, it follows that the competitive ratio is at least $\frac{1}{Z} \cdot (Z + k \frac{Z(Z+1)}{2}) = 1 + k \frac{Z+1}{2} = \Omega(kZ)$, as desired. This finishes the proof of \cref{thm:1}.

\section{Non-antimonotonicity of latency in delayed hits}
\label{sec:antimonotone}
In this section, we prove \cref{thm:2}. In fact, we will prove the following stronger lemma.
\begin{lemma}
\label{lem:thm2}
For $Z \geq 5$, there is a request sequence $\sigma$, and hit sequences $b,b'$ feasible with respect to $\sigma$ with $b' \geq b$ and $b'_t = b_t$ for all but one $t \in [T]$ such that $\ell_{\sigma}(b') - \ell_{\sigma}(b) = \Omega(Z^2)$ and $\ell_{\sigma}(b) = \opt(\sigma)$, that is, $\ell_{\sigma}(b)$ is the minimal possible latency for $\sigma$.
\end{lemma}
Note that, in particular, this implies \cref{thm:2} as it shows that $\ell_{\sigma}$ is not antimonotone. \cref{lem:thm2} shows that there is a scenario in which an algorithm has the option to have one additional cache hit if it wants, and it is optimal to not have the additional cache hit.

\begin{proof}
Let $z = \lfloor Z/2 \rfloor$. As a building block, we first consider the following request sequence $\sigma$: $i_1 = k+1$ and $i_{Z-z+1} = \ldots = i_{Z} = k+1$. Any algorithm $\A$ with input $\sigma$ will produce the hit sequence $b := (0, \dots, 0)$ and incur a latency of $\ell_{\sigma}(b) = Z + \sum_{r = 1}^{z} r = Z + z(z+1)/2$. Let $b' := (1, 0, \dots, 0)$. Observe that now, we have that $\ell_{\sigma}(b') = 0 + \sum_{r = 1}^z (Z - r + 1) = (Z + 1)z - z(z+1)/2$, and so we have that $\ell_{\sigma}(b) < \ell_{\sigma}(b')$ for $Z \geq 5$. Moreover, since $b$ is the only feasible hit sequence, we have that $\ell_{\sigma}(b)$ is the global minimum.

We now use the building block to finish the proof. For simplicity we will first assume that the cache size $k$ is $1$, and then show how to generalize the proof to larger $k$.

We modify $\sigma$ to define $\sigma'$ as follows. We first set $i_1 = k+1$, $i_2 = k+2$. Then, we play out $\sigma$ starting at $t = Z+1$, so $i_{Z+1} = k+1$, $i_{2Z - z + 1} = \ldots = i_{2Z} = k+1$. Finally, we set $i_{3Z + 1} = \ldots = i_{4Z} = k+2$. The modifications force the optimal algorithm to do the following. First, the request $i_1$ allows the algorithm to cache $k+1$ by time $t = Z+1$. The requests for item $k+2$ are there so that $k+1$ must be evicted by the cache in the retrieval phase of $t = Z+1$, or else the algorithm will incur a latency of $Z(Z+1)/2$ from the $Z$ requests for $k+2$ at the end. These combinations force the optimal algorithm to only need to decide whether or not to make the request at time $t = Z+ 1$ be a cache hit, and the building block earlier will show that it is better to not cache it.

In more detail, we observe that the hit sequence $b$, where $b_t = 0$ for all $t \leq 2Z$, and $b_{3Z+1} = \dots = b_{4Z} = 1$, is feasible. This is because this is the hit sequence for the algorithm that caches the first request for $k+2$ and otherwise does not modify the cache. We also observe that the hit sequence $b'$, where $b'_t = b_t$ except $b'_{Z+1} = 1$, is also feasible. This is because it is possible to cache the request $k+1$ at time $0$, use it to have a cache hit for the request at time $Z+1$, evict $k+1$ in $t = Z+1$'s retrieval phase when the request for $k+2$ comes back, and then never modify the cache again. Our building block shows that the difference in latencies is 
\begin{equation*}
\ell_{\sigma}(b') - \ell_{\sigma}(b) = \Big((Z+1)z - z(z+1)/2\Big) - \Big(Z + z(z+1)/2\Big) = z(Z - z) - Z \enspace,
\end{equation*}
 which is strictly greater than $0$ for $Z \geq 5$.

It remains to show that the hit sequence $b$ is optimal. Every algorithm must have a cache miss for the requests $i_1$ and $i_2$. Any algorithm which has a cache miss for $i_{3Z}$ incurs a latency of at least $Z(Z+1)/2$, because the request $i_{3Z}$ occurs $Z$ timesteps after $i_{2Z}$. This is greater than $\ell_{\sigma}(b)$ for $Z \geq 3$, so it follows that the optimal algorithm must cache $i_2$ during the retrieval phase at time $t = Z+1$, and never change the cache after that. It follows that the only potentially optimal realizable hit sequences are $b$ and $b'$, but we already know that $\ell_{\sigma}(b) < \ell_{\sigma}(b')$, so $b$ is optimal.

Finally, we explain how to modify the proof to work for $k > 1$. We modify $\sigma$ and add a sequence of $Z$ requests for item $i$ for $i \in \{2, \dots, k\}$ between the sequence of $z$ requests for $k+1$ and the sequence of $Z$ requests for $k+2$. As before, we space the sequences so that they are at least $Z$ timesteps apart. This forces the optimal algorithm to have items $2, \dots, k$ in the cache at time $t = 2Z$, as if it ever evicts item $i$ it will incur a latency of at least $Z(Z+1)/2$ from the sequence of $Z$ requests for item $i$. This forces the optimal caching algorithm to only be able to evict item $1$, which then reduces to the case where $k = 1$.
\end{proof}

\section{Antimonotone delayed hits}
\label{sec:model}
In this section, we prove \cref{thm:3}. In \cref{sec:antimodel} we define the antimonotone delayed hits model and show that its latency function is antimonotone, and in \cref{sec:thm3proof} we prove \cref{thm:3}.

\subsection{The antimonotone delayed hits model}
\label{sec:antimodel}
The antimonotone delayed hits model is a simple modification to the delayed hits model. We use our intuition from \cref{lem:thm2} to motivate the modification. In the proof of \cref{lem:thm2}, we constructed a request sequence $\sigma$ where when a particular request misses, it decreases the latencies of $Z/2$ requests by $Z/2$ \emph{each}, resulting in an overall decrease in latency. The issue is that if the request were to hit, then we would not fetch the item from the backing store, so then there is no request for the item ``in flight'' when the first of the $Z/2$ requests arrives. This issue is fixed by fetching items from the backing store \emph{even when there is a cache hit}, so that way it is never advantageous to have a cache miss.

Formally, the modification is as follows. Before, when we had a request $i_t$ that was a cache hit we would not send this request to the backing store; now we do.\footnote{In a real-world system this results in always sending requests to the backing store even if the request is in the cache. Whether or not this is realistic in practice depends on whether or not the backing store has the additional throughput to handle the extra requests.} This results in the following latency function.
\begin{fact}[Latency function for antimonotone delayed hits]
The latency function for the antimonotone delayed hits model is as follows. For a request sequence $\sigma = (i_t)_{t = 1}^T$, we let $\ell'_{\sigma} \colon \Bits^T \to \R$ be the function $\sum_{t = 1}^T l_t$ where $l_t$ is defined to be $l_t = (1 - b_t)(Z - t + \min_{t' \in [t - Z + 1, t] : i_{t'} = i_t} t')$.
\end{fact}
The explicit definition of the latency function is nearly identical to one for the delayed hits model that was uncovered in the proof of \cref{lem:latencyfn}. The only difference is now we $p_t = \min_{t' \in [t - Z + 1, t] : i_{t'} = i_t} t'$ instead of $\min_{t' \in [t - Z + 1, t] : b_{t'} = 0 \wedge i_{t'} = i_t} t'$. This is because the request made at time $t'$ can now be used to serve the request at time $t$ even when $b_{t'} = 1$, i.e.\ even when the request at time $t'$ was a cache hit. As we shall see, this removes the dependencies of the latency function on the hit sequence that caused the delayed hits latency function to be antimonotone.

\begin{claim}
For every $\sigma$, $\ell'_{\sigma}$ is antimonotone.
\end{claim}
\begin{proof}
Fix $b, b' \in \Bits^T$ with $b' \geq b$. Let $S \subseteq [T]$ be the set of $t$ where $b'_t = b_t$. Note that when $t \notin S$ we must have $b'_t = 1$ and $b_t = 0$. Let $p_t = \min_{t' \in [t - Z + 1, t] : i_{t'} = i_t} t'$. We have that $\ell'_{\sigma}(b) = \sum_{t} l_t$ and $\ell'_{\sigma}(b') = \sum_{t} l'_t$ where $l_t = (1 - b_t)(Z - t + p_t)$ and $l'_t = (1 - b'_t)(Z - t + p_t)$. Hence,
\begin{flalign*}
&\ell'_{\sigma}(b) - \ell'_{\sigma}(b') = \sum_{t = 1}^T (l_t - l'_t) = \sum_{t = 1}^T (b'_t - b_t)(Z - t + p_t)  \\
&= 0 + \sum_{t \notin S} (b'_t - b_t)(Z - t + p_t) = \sum_{t \notin S} (Z - t + p_t) \geq 0 \enspace,
\end{flalign*}
since $t - p_t \leq Z - 1$. Therefore, $\ell'_{\sigma}$ is antimonotone.
\end{proof}

\subsection{Proof of \cref{thm:3}}
\label{sec:thm3proof}
Let $\A$ be an algorithm in the antimonotone delayed hits model that uses a cache of size $k$. We give an algorithm $\B$ in the delayed hits model that uses a cache of size at most $k + Z$ such that $\text{latency}(\A)$ in the antimonotone delayed hits model is at least $\text{latency}(\B)$ in the delayed hits model. 

The algorithm $\B$ splits its cache into two components: $S_0$ and $S_1$, where $S_0$ is a cache of size $k$ and $S_1$ is a cache of size $Z$. The algorithm $\B$ then simulates the algorithm $\A$, using $S_0$ to maintain $\A$'s cache. Additionally, the algorithm $\B$ uses $S_1$ to store the last $Z$ requests returned from the backing store. For example, when $\A$ evicts an item from its cache, $\B$ will keep the item in $S_1$ if it was one of the last $Z$ requests that it has seen.

We now show that $\B$'s latency is at most $\A$'s latency. In fact, we show that the latency experienced by any request $i_t$ is can only be lower for $\B$ than for $\A$. This is because the only difference between the antimonotone delayed hits model and the delayed hits model is that we send every request to the backing store, irrespective of whether or not it is a cache hit. So, in the antimonotone delayed hits model is it possible for the item to have a lower latency than in the delayed hits model, but only if it was one of the last $Z$ items requested. But in this case it will be in $\B$'s cache, and so it will have a latency of $0$ for algorithm $\B$, which is at most the latency it has in algorithm $\A$, which finishes the proof.

\section{Acknowledgements}
The authors thank Nirav Atre for helpful discussions, and Magdalen Dobson for providing helpful comments on an earlier draft of the manuscript.

\newpage
\bibliography{references}
\bibliographystyle{alpha}
\end{document}